\documentclass[journal]{IEEEtran}
\usepackage{amsthm}
\usepackage{amssymb}
\usepackage{graphicx}
\usepackage{amsmath}
\usepackage{newtxmath}
\usepackage[cal=cm]{mathalpha} 
\usepackage{tablefootnote} 
\usepackage{threeparttable} % 在导言区添加
\usepackage{bm}
\newtheorem{theorem}{Theorem}{}
{}
{}
{}
\newtheorem{corollary}{Corollary}{}
{}

\usepackage{algorithm}
\usepackage{algpseudocode}
\usepackage{graphics}
\usepackage{epsfig}
\usepackage{multirow}
\usepackage{caption}
\usepackage{array}
\usepackage{cite}
\usepackage{stfloats}
\usepackage{subfig} % 在导言区添加
\usepackage{booktabs}
\usepackage{comment}
\usepackage{color}
\usepackage{tabu} 
\usepackage{stfloats}
\usepackage{mathtools}
\usepackage{stmaryrd} %CP Decomposition, \llbracket \rrbracket 
\usepackage[colorlinks,linkcolor=blue,anchorcolor=blue,citecolor=blue,bookmarks=true]{hyperref}
\usepackage[capitalize]{cleveref} % 导言区
\crefname{figure}{Fig.}{Figs.}   % 小写引用：\cref → "Fig. 1", "Figs. 1 and 2"
\Crefname{figure}{Fig.}{Figs.}   % 大写引用：\Cref → "Fig. 1"（句首时）
\usepackage{subcaption}
\makeatletter
\let\sum\relax
\let\prod\relax
\DeclareSymbolFont{largesymbols}{OMX}{cmex}{m}{n}
\DeclareMathSymbol{\sum}{\mathop}{largesymbols}{"50}
\DeclareMathSymbol{\prod}{\mathop}{largesymbols}{"51}
\DeclareMathOperator*{\diag}{\text{diag}}

\begin{document}
	\title{Rotatable Antenna-Enhanced Wireless Sensing with Uniform Sparse Array via Tensor Decomposition}
	
	
	\author{\IEEEauthorblockN{
			Chengzhi Ye, \IEEEmembership{Student Member,~IEEE}, 
			Ruoyu Zhang, \IEEEmembership{Senior Member,~IEEE}, 
			Jincheng Du, %Guangyi Chen,
			Wenyan Ma,
			\IEEEmembership{Member,~IEEE},
			Qingqing Wu, 
			\IEEEmembership{Senior Member,~IEEE},
			Wen Wu, 
			\IEEEmembership{Senior Member,~IEEE},
			Rui Zhang,
			\IEEEmembership{Fellow,~IEEE}
		}
		\vspace{-0.25em}
		\thanks{\vspace{-1em}
			%This work was supported in part by the National Natural Science Foundation of China under Grant 62571248, Grant 62201266, and in part by the Undergraduate Research Training Program of Nanjing University of Science and Technology under Grant 202510288068.
			
			Chengzhi Ye, Ruoyu Zhang, Jincheng Du and Wen Wu are with the Key Laboratory of Near-Range RF Sensing ICs and Microsystems (NJUST), Ministry of Education, School of Electronic and Optical Engineering, Nanjing University of Science and Technology, Nanjing 210094, China. (e-mail: qxsycz8166@njust.edu.cn; ryzhang19@njust.edu.cn; dujincheng@njust.edu.cn;
			wuwen@njust.edu.cn). Wenyan Ma and Rui Zhang are with the Department of Electrical and Computer Engineering, National University of Singapore, Singapore 117583 (e-mail: wenyan@u.nus.edu; elezhang@nus.edu.sg). Qingqing Wu is with the Department of Electronic Engineering, Shanghai Jiao Tong University, Shanghai 200240, China (email: qingqingwu@sjtu.edu.cn)
			\textit{(Corresponding author: Ruoyu Zhang)}}
		%\thanks
		\vspace{-2em}
	}

	
	\maketitle
	
	%\begin{abstract}
	%f
	%\end{abstract}
	
	%% no keywords
	%\begin{IEEEkeywords}
	%Sparse vector coding, 
	%\end{IEEEkeywords}

	% \fi
	\IEEEpeerreviewmaketitle

	\begin{abstract}
		In this letter, we propose a new wireless sensing system equipped with a rotatable antenna (RA) array to enhance the sensing performance of a uniform sparse array (USA). To tackle the severe spatial undersampling issues, we propose a novel tensor decomposition-based direction-of-arrival (DOA) estimation algorithm. Specifically, a synchronous multi-rotation pattern is designed for active target probing, allowing the received signals collected over multiple rotations to capture diverse spatial degrees of freedom. Subsequently, we mathematically formulate the received signals across successive rotations as a third-order tensor, and leverage the canonical polyadic decomposition to obtain the factor matrices incorporating the DOA of targets. By analyzing the extrema distribution laws of the array steering vector correlation (SVC) and the gain SVC of RAs, we further combine the array and gain factor matrices via the Kronecker product, which theoretically guarantees unambiguous DOA estimation. Simulation results demonstrate that the proposed RA-enhanced tensor decomposition-based algorithm achieves high-precision and ambiguity-free sensing performance compared to conventional uniform dense arrays and omnidirectional antenna systems.
	\end{abstract}
	\begin{IEEEkeywords}
		Wireless sensing, rotatable antenna, uniform sparse array
	\end{IEEEkeywords}
	\vspace{-0.5em}
	\section{Introduction}
	
	The forthcoming sixth-generation (6G) networks are anticipated to enable a massive proliferation of location-aware applications such as autonomous driving and robotic navigation. These advanced use cases require high-precision sensing capabilities from the wireless infrastructure. Consequently, wireless sensing is widely envisioned to become a primary service to facilitate the parameter estimation of targets \cite{DOAHAD2022Ruoyu}.

	To achieve superior sensing performance with high angular resolution, large-scale arrays are typically deployed in modern wireless systems \cite{LiQiangAccurate2019}. To reduce the implementation cost while maintaining a comparable spatial resolution, sparse antenna arrays are widely utilized by enlarging the distance between adjacent elements with significantly fewer antennas \cite{LiuSuperNestedI2016}. %,LiuSuperNestedII2016
	Furthermore, movable antennas (MAs) have recently been introduced into wireless sensing to intelligently adjust the physical positions of antennas. This approach can effectively expand the spatial array aperture within a given region to further improve the sensing precision \cite{maMovableAntennaEnhanced2024, ChengzhiIOTJ20262D}.
	
	Nevertheless, the performance gains of large-scale arrays are inevitably accompanied by prohibitive hardware costs and immense computational burdens \cite{ISACRZzhang2024,han2020channel}. Although sparse arrays offer an economical alternative, the most straightforward uniform sparse array (USA) suffers from severe spatial undersampling, which inevitably generates grating lobes, introduces severe signal interference, and causes estimation ambiguities \cite{PatraDOANested2025}. These inherent physical limitations fundamentally restrict the practical sensing performance and strongly necessitate the development of more advanced array architectures. To mitigate these issues, existing works are forced to impose strict restrictions on antenna element placements, such as nested or coprime arrays \cite{ VaidyanathanCoPrime2011}. On the other hand, while MAs can flexibly expand the aperture, their limited moving speed results in degraded real-time sensing capabilities and low observation efficiency \cite{ZhangChannel2024,DOAGPEChengzhi2025}.
	
	Recently, rotatable antennas (RAs), as a simplified realization of six-dimensional MA \cite{10752873}, have demonstrated tremendous potential to address these hardware challenges. An RA typically employs a directional antenna capable of concentrating signal energy towards a specific orientation to achieve highly focused sensing and transmission \cite{zheng2026rotatableantenna}. Furthermore, an RA can rapidly adjust its radiation direction via motor controls or electronic switches \cite{Zheng2025RotatableOpportunities}. This instantaneous switching mechanism significantly improves the temporal efficiency and instantaneous capabilities of environmental perception. Driven by these remarkable advantages, existing studies have actively explored the application of RAs in wireless communications, channel estimations, and integrated sensing and communication systems \cite{ZhengRotatableComm2026, XiongChannel2025}. Despite these remarkable advantages, integrating RAs with USAs to fundamentally overcome the inherent grating lobe issue by leveraging their directional filtering has not been explored.

	To tackle the aforementioned challenges and fully unleash the sensing potential of RAs, we propose a novel tensor decomposition-based wireless sensing method for RA-enhanced USA sensing systems. Specifically, we design a synchronous multiple rotation pattern to capture diverse spatial degrees of freedom such that the received signals across successive rotations can be mathematically formulated as a third-order tensor. Subsequently, we employ the canonical polyadic (CP) decomposition to extract the factor matrices. By investigating the array and the gain factor matrix, we analytically prove that the strict unimodal property of the gain steering vector correlation (SVC) can perfectly compensate for the inherent spatial ambiguity of the array SVC. Driven by this theoretical analysis, we propose to combine the gain and array factor matrices via the Kronecker product, which achieves an unambiguous DOA estimation. Finally, simulation results are provided to demonstrate that the proposed RA-enhanced tensor decomposition-based algorithm achieves robust, high-precision, and unambiguous sensing performance, which significantly outperforms the conventional uniform dense array (UDA) and omnidirectional antenna (OA) systems.
	
	%	Notations: Symbols for vectors (lower case) and matrices (upper case) are in boldface. The transpose and conjugate transpose are represented by $(\cdot)^{\mathrm{T}}$, $(\cdot)^{\mathrm{H}}$, respectively.
	%	The covariance is given by $\mathbb{E}\{\cdot\}$. The $K\times K$-dimensional identity matrix and $K\times L$-dimensional zero matrix are expressed by $\bm{I}_{K\times K}$ and $\bm{0}_{K\times L}$. The symbols $\circledast$ and $\dagger $ denote the Hadamard product and pseudoinverse, respectively. The $T$-dimensional vector with all elements equal to 1 is written as $\bm{1}_{T \times 1}$. The expression $[\bm{\alpha}]_k$ denotes the $k$-th element of $\bm{\alpha}$. The expression $[\bm{B}]_{m,:}$ and $[\bm{B}]_{:,k}$
	%	denote the $m$-th row and the $k$-th column of
	%	$\bm{B}$, respectively. The expression $\rank[\bm{B}]$ denotes the column rank of $\bm{B}$.
%	\vspace{-0.5em}
	\section{System Model}
	\begin{figure}[t]
		\centering
		\includegraphics[width=0.8\linewidth]{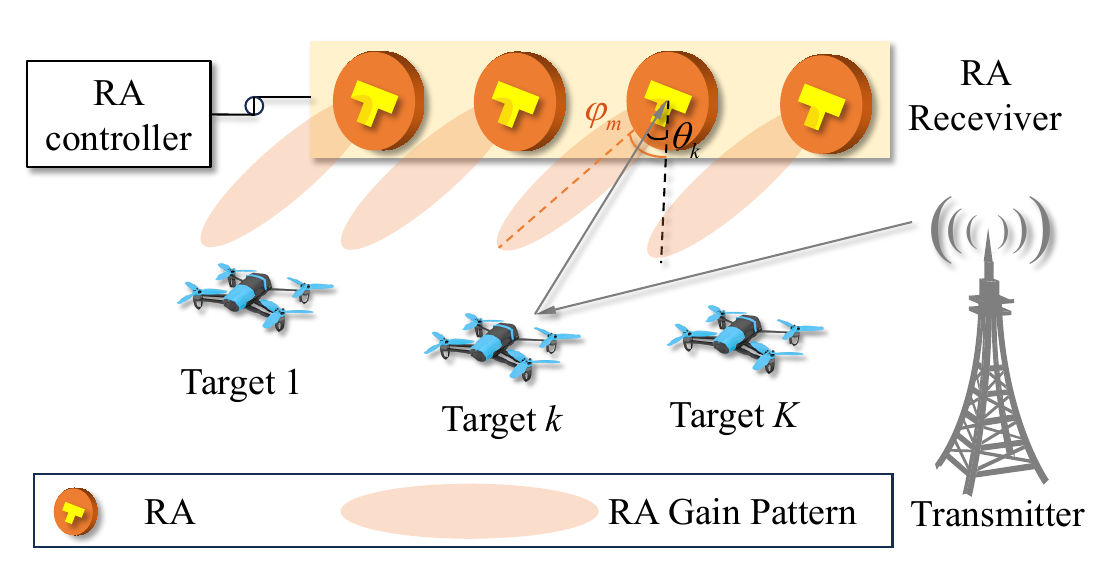}
		\caption{The RA system with USA.}
		\vspace{-0.75em}
		\label{FigRA}
	\end{figure}
	
	We consider a wireless sensing model where the transmitter is equipped with an omnidirectional antenna to transmit periodic waves for target sensing. The receiver is equipped with $N$ RAs to sense $K$ uncorrelated targets located in the far field. The DOA of the $k$-th target is denoted by $\theta_k$. Without loss of generality, we assume that the targets are located within a sensing range bounded by $[-\vartheta_{\max},\vartheta_{\max}]$, where $\vartheta_{\max}$ represents the maximum elevation angle sensed by the RA receiver. As shown in \Cref{FigRA}, the array is equipped with an electronic controller that can simultaneously and synchronously adjust the directivity of the $N$ RAs while keeping their physical positions fixed through electronic technologies such as reconfigurable parasitic radiator loading or diode switches \cite{zheng2026rotatableantenna}. For simplicity, we define the angle between the RA pointing direction and the array normal direction as the ``rotation angle''. We assume that the RAs rotate $M$ times to maximally cover the sensing area of the active probing, and the rotation angle vector can be expressed as $\bm{\varphi}=[\varphi_1, \ldots, \varphi_M]^{\mathrm{T}} \in \mathbb{R}^{M \times 1}$, where $\varphi_m = -\vartheta_{\max}+\frac{2(m-1)\vartheta_{\max}}{M-1}$.
	
	The effective antenna gain of the RA depends on the signal angle of arrival and the rotation angle. In this paper, we adopt the following widely used directional gain pattern model. During the $m$-th rotation, the power gain for the $k$-th target can be expressed as \cite{AntennaTheory}
	\begin{align}
		g_{m,k}(\theta_k, \varphi_m) = \begin{cases} G\cos^{2p}(\theta_k - \varphi_m), \!\!& |\theta_k - \varphi_m| \le \pi/2, \\ 0, & \text{otherwise,} \end{cases}
	\end{align}
	where $p \ge 0$ represents the directivity factor of the RA main lobe beamwidth, $G = 2(2p+1)$ is the maximum gain in the RA pointing direction satisfying the law of power conservation, and $\theta_k$ is the DOA of the $k$-th target. The RA adopts a USA structure in order to enlarge the array aperture, whose array steering vector for the $k$-th target can be expressed as
	\begin{align}
		\bm{a}_k(\theta_k, L)= \big[1, e^{j \frac{2\pi}{\lambda}Ld\sin \theta_k}, \ldots, e^{j \frac{2\pi}{\lambda}(N -1)Ld \sin \theta_k} \big]^{\mathrm{T}},
	\end{align}
	where $d=\frac{\lambda}{2}$ represents the distance between RAs and $L\ge 1$ denotes the sparse factor. %The case of $L=1$ corresponds to a conventional uniform dense array (UDA), and $L>1$ corresponds to a USA. 
	Therefore, the sensing channel of the $k$-th target during the $m$-th rotation can be formulated as
	\begin{align}
		\bm{h}_{m,k}(\theta_k,{\varphi}_m) =\alpha_k b_{m,k}(\theta_k,{\varphi}_m)\bm{a}_k(\theta_k, L) ,
	\end{align}
	%where $\alpha_k=\sqrt{\frac{\varOmega_k \lambda^2}{(4\pi)^3 R_k^4}}$ denotes the scattering coefficient of the $k$-th target, and $\varOmega_k$ and $R_k$ represent the radar cross-section of the $k$-th target and the LOS distance from the target to the receiver respectively.
	where $\alpha_k$ denotes the scattering coefficient of the $k$-th target and $b_{m,k} = \sqrt{g_{m, k}(\theta_k, \varphi_m)}$. Thus, the sensing channel of the $K$ targets during the $m$-th rotation can be expressed as
	\begin{align}
		\bm{H}_m(\bm{\theta},{\varphi}_m) = \bm{A}(\bm{\theta}, L)\bm{\varGamma}_m(\bm{\theta}, {\varphi}_m)\bm{\varLambda} \in \mathbb{C}^{N \times K},
	\end{align}
	where $\bm{\varGamma}_m(\bm{\theta},{\varphi}_m) = \diag(\bm{b}_m(\bm{\theta}, \varphi_m)) \in\mathbb{C}^{K \times K}$, $\bm{b}_m(\bm{\theta}, \varphi_m) = [b_{m, 1}, \ldots,b_{m, K}]^{\mathrm{T}}$, $\bm{A} = [\bm{a}_1(\theta_1, L), \ldots, \bm{a}_K(\theta_K, L)] \in \mathbb{C}^{N\times K}$, and $\bm{\varLambda}=\diag(\alpha_1, \ldots, \alpha_K) \in\mathbb{C}^{K \times K}$. We assume that the transmitter sends the same probing signal during each rotation of the RAs. Consequently, the received signal at the base station during the $t$-th snapshot can be given by
	\begin{align}
		\bm{y}_m(t) = \bm{H}_m(\bm{\theta},{\varphi}_m)\bm{\tilde{s}}(t) + \bm{n}_m(t),
	\end{align}
	where $\bm{\tilde{s}}(t) = [\tilde{s}_1(t), \ldots,\tilde{s}_K(t)]^{\mathrm{T}}$ represents the signal at the $t$-th snapshot and satisfies $\mathbb{E}\{\bm{\tilde{s}}(t)\bm{\tilde{s}}^{\mathrm{H}}(t)\} = \diag(\sigma_1^2, \ldots, \sigma_K^2)$. $\sigma_k^2$ represents the power of the $k$-th signal. $\bm{n}_m(t)\sim\mathcal{C}\mathcal{N}(0, \sigma_n^2\bm{I}_{N})$ denotes the additive white Gaussian noise with $\sigma_n^2$ being the noise power. By collecting $T$ snapshots, the received signal during the $m$-th rotation can be written as
	\begin{align}
		\bm{Y}_m = \bm{A}(\bm{\theta}, L) \bm{\varGamma}_m(\bm{\theta},{\varphi}_m) \bm{S}+\bm{N}_m,
	\end{align}
	where $\bm{Y}_m = [\bm{y}_m(1), \ldots, \bm{y}_m(T)] \in \mathbb{C}^{N \times T}$, $\bm{S} = \bm{\varLambda}\bm{\tilde{S}}$, $\bm{\tilde{S}} = [\bm{\tilde{s}}(1), \ldots, \bm{\tilde{s}}(T)] \in \mathbb{C}^{K \times T}$, and $\bm{N}_m = [\bm{n}_m(1), \ldots, \bm{n}_m(T)] \in \mathbb{C}^{N \times T}$.
	In the next section, we discuss how to model the received signals as a tensor and perform CP decomposition and DOA estimation.
	
	\section{Tensor Decomposition and DOA Estimation}
	
	In this section, we model the received signals as a tensor and employ the alternating least squares (ALS) algorithm to perform CP decomposition for obtaining the estimation of the factor matrices. Subsequently, we analyze the extrema distribution laws of the array SVC and the gain SVC. Based on this, we propose to combine the array and gain factor matrices via the Kronecker product to achieve the unambiguous DOA estimation.
	\vspace{-0.5em}
	\subsection{Tensor Modeling}
	
	We vertically concatenate the received signals from the $M$ rotations to obtain the total received signal matrix $\bm{Y}\in \mathbb{C}^{MN \times T}$, which can be expressed as \eqref{YTensor}, where $\bm{B}(\bm{\theta},\bm{\varphi}) = \big[\bm{b}_1(\theta_1, \bm{\varphi}), \ldots, \bm{b}_K(\theta_K, \bm{\varphi})\big]\in \mathbb{C}^{M \times K}$ and $\bm{b}_k(\theta_k, \bm{\varphi}) = [b_{1, k}, \ldots, b_{M,k}]^{\mathrm{T}}$ is the gain steering vector. The last equality holds based on the definition of the Khatri-Rao product.
	
	\begin{figure*}
		\begin{align}\label{YTensor}
			\!\!\bm{Y} \!= \!
			\begin{bmatrix}
				\bm{Y}_1 \\ \bm{Y}_2 \\ \vdots\\ \bm{Y}_M
			\end{bmatrix}
			\!=\!
			\begin{bmatrix}
				\bm{A}(\bm{\theta}, L) \bm{\varGamma}_1(\bm{\theta},{\varphi}_1) \\ \bm{A}(\bm{\theta}, L) \bm{\varGamma}_2(\bm{\theta},{\varphi}_2) \\ \vdots\\ \bm{A}(\bm{\theta}, L)\bm{\varGamma}_M(\bm{\theta},{\varphi}_M)
			\end{bmatrix}\bm{{S}}
			\!+\!
			\begin{bmatrix}
				\bm{N}_1 \\ \bm{N}_2 \\ \vdots\\ \bm{N}_M
			\end{bmatrix}
			\!=\!
			\begin{bmatrix}
				b_{1,1}\bm{a}_1&b_{1,2}\bm{a}_2&\ldots&b_{1,K}\bm{a}_K \\ 
				b_{2,1}\bm{a}_1&b_{2,2}\bm{a}_2&\ldots&b_{2,K}\bm{a}_K\\
				\vdots&\vdots& \ddots&\vdots\\
				b_{M,1}\bm{a}_1&b_{M,2}\bm{a}_2&\ldots&b_{M,K}\bm{a}_K
			\end{bmatrix}\bm{{S}}+\bm{N} = \big[\bm{B}(\bm{\theta},\bm{\varphi}) \odot\bm{A}(\bm{\theta}, L) \big] \bm{{S}}+\bm{N}.
		\end{align}
		\vspace{-0.5em}
		\hrule
		\vspace{-1em}
	\end{figure*}
	
	According to the theory of tensor CP decomposition \cite{ZhangChannel2024}, $\bm{Y}$ can be viewed as the matrix form of a third-order tensor $\bm{\mathcal{Y}} \in\mathbb{C}^{N\times M\times T}$, which can be specifically represented as the sum of $K$ rank-one tensors
	\begin{align}\label{YTen}
		\bm{\mathcal{Y}} &= \sum_{k=1}^{K}{\bm{a}_k(\theta_k, L) \circ \bm{b}_k(\theta_k, \bm{\varphi}) \circ \bm{s}_k} + \bm{\mathcal{N}}
		\\&= \llbracket \bm{A}(\bm{\theta}, L), \bm{B}(\bm{\theta},\bm{\varphi}),\bm{S}^{\mathrm{T}}\rrbracket + \bm{\mathcal{N}},\nonumber
	\end{align}
	where the $(n, m, t)$-th element of $\bm{\mathcal{Y}}$ corresponds to the element at the $\big((m-1)N+n\big)$-th row and the $t$-th column of the matrix $\bm{Y}$. $\bm{s}_k \in \mathbb{C}^{T \times 1}$ is the $k$-th column of $\bm{S}^{\mathrm{T}}$. $\bm{\mathcal{N}}\in\mathbb{C}^{N\times M\times T}$ represents the noise tensor. We refer $\bm{A}(\bm{\theta}, L)$, $\bm{B}(\bm{\theta},\bm{\varphi})$, and $\bm{S}^{\mathrm{T}}$ to the array, gain, and signal factor matrices, respectively.
	\vspace{-0.5em}
	\subsection{Tensor Decomposition}
	
	Our objective is to estimate the unknown target parameters $\{\theta_k\}_{k=1}^{K}$. Based on \eqref{YTen}, the parameter estimation problem can be modeled as
	\begin{align}
		\min_{\{\theta_k\}_{k=1}^{K}} \enspace \Big\| \bm{\mathcal{Y}} -\sum_{k=1}^{K}{\bm{a}_k(\theta_k, L) \circ \bm{b}_k(\theta_k, \bm{\varphi}) \circ \bm{s}_k} \Big\|_{\mathrm{F}}^2.
	\end{align}
	There are various methods to solve this CP decomposition problem. A commonly used approach is the ALS algorithm. During the iteration process, one factor matrix is updated by alternately fixing the other factor matrices until convergence is reached. The updated factor matrices are given by
	\begin{align}
		\bm{A}_{(\iota+1)} &= \arg \min_{\bm{A}} \Big\|\bm{Y}_{(1)}-\bm{A}\Big( \bm{S}_{(\iota)}^{\mathrm{T}} \odot \bm{B}_{(\iota)}\Big)^{\mathrm{T}} \Big\|_{\mathrm{F}}^2, \\
		\bm{B}_{(\iota+1)} &= \arg \min_{\bm{B}} \Big\|\bm{Y}_{(2)}-\bm{B}\Big( \bm{S}_{(\iota)}^{\mathrm{T}} \odot \bm{A}_{(\iota+1)}\Big)^{\mathrm{T}} \Big\|_{\mathrm{F}}^2, \\
		\bm{S}_{(\iota+1)}^{\mathrm{T}} &= \arg \min_{\bm{S}^{\mathrm{T}}} \Big\|\bm{Y}_{(3)}-\bm{S}^{\mathrm{T}}\Big( \bm{B}_{(\iota+1)} \odot \bm{A}_{(\iota+1)}\Big)^{\mathrm{T}} \Big\|_{\mathrm{F}}^2,
	\end{align}
	where $\bm{Y}_{(1)} \in \mathbb{C}^{N \times MT}$, $\bm{Y}_{(2)}\in \mathbb{C}^{M \times NT}$, and $\bm{Y}_{(3)}\in \mathbb{C}^{T \times MN}$ denote the mode-1, mode-2, and mode-3 unfoldings of $\bm{\mathcal{Y}}$, respectively. The terms $\bm{A}_{(\iota+1)}$, $\bm{B}_{(\iota+1)}$, and $\bm{S}_{(\iota+1)}^{\mathrm{T}}$ represent the factor matrices updated at the $\iota$-th iteration. By alternately and iteratively performing the above three updates, we can obtain the estimates of the array, gain, and signal factor matrices $\{\bm{\hat{A}},  \bm{\hat{B}},\bm{\hat{S}}^{\mathrm{T}}\}$. Note that the factor matrices obtained through CP decomposition typically contain errors and are subject to scaling and permutation ambiguities, which can be given by
	\begin{align}
		\bm{\hat{A}} &= \bm{A}(\bm{\theta}, L)\bm{\varPi}\bm{\varDelta}^{(1)} + \bm{E}^{(1)}, \label{AAhat}\\
		\bm{\hat{B}} &= \bm{B}(\bm{\theta},\bm{\varphi}) \bm{\varPi}\bm{\varDelta}^{(2)} + \bm{E}^{(2)}, \label{BBhat}\\
		\bm{\hat{S}}^{\mathrm{T}} &= \bm{S}^{\mathrm{T}} \bm{\varPi}\bm{\varDelta}^{(3)} + \bm{E}^{(3)},
	\end{align}
	where $\bm{\varPi} \in \mathbb{C}^{K \times K}$ denotes the permutation matrix, while $\bm{\varDelta}^{(i)} \in \mathbb{C}^{K \times K}$ and $\bm{E}^{(i)}$, for $i = 1, 2,3$, represent the diagonal scaling matrices and the tensor decomposition error matrices, respectively. It is worth noting that even under noise-free conditions, the estimated factor matrices $\{\bm{\hat{A}},  \bm{\hat{B}},\bm{\hat{S}}^{\mathrm{T}}\}$ are not identical to the original factor matrices $\{\bm{A}(\bm{\theta}, L), \bm{B}(\bm{\theta},\bm{\varphi}), \bm{S}^{\mathrm{T}}\}$. However, each column of an estimated factor matrix corresponds to the information of the same target, which does not affect the DOA estimation. In the next subsection, we demonstrate how to estimate the DOA based on the sparse array using the estimated factor matrices.
	\vspace{-0.5em}
	\subsection{DOA Estimation}
	
	Through the aforementioned tensor CP decomposition, $\bm{\theta}$ is implicitly coupled in the array and gain factor matrices $\bm{A}(\bm{\theta}, L)$ and $\bm{B}(\bm{\theta},\bm{\varphi})$. According to \eqref{AAhat} and \eqref{BBhat}, we obtain
	\begin{align}
		\bm{\hat{a}}_k &= \delta_k^{(1)}\bm{a}_k(\theta_k, L)+\bm{e}_k^{(1)}\label{aahat},\\
		\bm{\hat{b}}_k &= \delta_k^{(2)}\bm{b}_k(\theta_k, \bm{\varphi})+\bm{e}_k^{(2)}\label{bbhat},
	\end{align}
	where $\delta_k^{(i)}$ is the $k$-th diagonal element of $\bm{\varDelta}^{(i)}$, and $\bm{e}_k^{(i)}$ is the $k$-th column of $\bm{E}^{(i)}$ for $i = 1,2$. Based on the array factor matrix $\bm{A}(\bm{\theta},L)$ and the gain factor matrix $\bm{B}(\bm{\theta}, \bm{\varphi})$, the steering vector correlation between the DOA of the $k$-th target $\theta_k$ and an arbitrary angle $\vartheta$ can be given respectively by
	\begin{align}\label{ASVC}
		\mathcal{A}(\theta_k, \vartheta, L)&\triangleq\frac{|\bm{a}_k^{\mathrm{H}}(\theta_k, L)\bm{a}(\vartheta, L)|^2}{||\bm{a}_k(\theta_k, L)||_2^2 ||\bm{a}(\vartheta, L)||_2^2},\\
		\mathcal{B}(\theta_k, \vartheta, \bm{\varphi})&\triangleq\frac{|\bm{b}_k^{\mathrm{T}}(\theta_k, \bm{\varphi})\bm{b}(\vartheta, \bm{\varphi})|^2}{||\bm{b}_k(\theta_k, \bm{\varphi})||_2^2 ||\bm{b}(\vartheta, \bm{\varphi})||_2^2}\label{BSVC},
	\end{align}
	where $\mathcal{A}(\theta_k, \vartheta, L)$ denotes the array SVC and $\mathcal{B}(\theta_k, \vartheta, \bm{\varphi})$ denotes the gain SVC.
	We further investigate them below. 
	
	Specifically, $\mathcal{A}(\theta_k, \vartheta, L)$ can be simplified as \eqref{SVCA}, shown on top of the next page
	\begin{figure*}
		\begin{align}
			\mathcal{A}(\theta_k, \vartheta, L) &= \frac{1}{N^2} \Big| \sum\nolimits_{n=0}^{N-1} \big( e^{-j \pi L n \sin\theta_k} \big) \big( e^{j \pi L n \sin\vartheta} \big) \Big|^2 = \frac{1}{N^2} \Big| \sum\nolimits_{n=0}^{N-1} e^{j n \pi L (\sin\vartheta - \sin\theta_k)} \Big|^2 = \frac{1}{N^2} \Big| \frac{1 - e^{j N \pi L (\sin\vartheta - \sin\theta_k)}}{1 - e^{j \pi L (\sin\vartheta - \sin\theta_k)}} \Big|^2 \nonumber \\
			&= \frac{1}{N^2} \Bigg| \frac{e^{j \frac{N \pi L}{2} (\sin\vartheta - \sin\theta_k)} \big[ e^{-j \frac{N \pi L}{2} (\sin\vartheta - \sin\theta_k)} - e^{j \frac{N \pi L}{2} (\sin\vartheta - \sin\theta_k)} \big]}{e^{j \frac{\pi L}{2} (\sin\vartheta - \sin\theta_k)} \big[ e^{-j \frac{\pi L}{2} (\sin\vartheta - \sin\theta_k)} - e^{j \frac{\pi L}{2} (\sin\vartheta - \sin\theta_k)} \big]} \Bigg|^2 = \frac{1}{N^2} \frac{\sin^2\big[ \frac{N \pi L}{2} (\sin\vartheta - \sin\theta_k) \big]}{\sin^2\big[ \frac{\pi L}{2} (\sin\vartheta - \sin\theta_k) \big]}  \label{SVCA}.
		\end{align}
		\vspace{-0.5em}
		\hrule
		\vspace{-0.75em}
	\end{figure*}
	where the third equality utilizes the sum formula of a geometric progression, and the forth equality applies Euler's formula. As a variant of the Dirichlet kernel, $\mathcal{A}(\theta_k, \vartheta, L)$ achieves the maximum value of 1 if and only if the following condition is satisfied
	\begin{align}
		\sin\vartheta = \sin\theta_k + {2z}/{L},\enspace z\in \mathbb{Z}.
	\end{align}
	Since $\theta_k,\vartheta \in [-\vartheta_{\max}, \vartheta_{\max}]$ and $|\vartheta|<\frac{\pi}{2}$, the equation has a unique solution $\vartheta =\theta_k$ if and only if $L =1$. When $L>1$, there are multiple values of $\vartheta$ that maximize $\mathcal{A}(\theta_k, \vartheta, L)$. Therefore, $\theta_k$ cannot be estimated unambiguously relying solely on the array factor matrix $\bm{A}(\bm{\theta}, L)$, which reveals the fundamental limitation of conventional sparse arrays.
	
	Next, we proceed to discuss the properties of $\mathcal{B}(\theta_k, \vartheta, \bm{\varphi})$.
	\begin{theorem}\label{Theorem1}
		When the number of rotations satisfies $M\ge 2$, $\mathcal{B}(\theta_k, \vartheta, \bm{\varphi})$ possesses a unique global maximum within the sensing range if and only if $\vartheta = \theta_k$.
	\end{theorem}
	\begin{proof}
		According to the Cauchy-Schwarz inequality, the inner product of two vectors is maximized when they are collinear. When $\bm{b}(\theta_k, \bm{\varphi})$ and $\bm{b}(\vartheta, \bm{\varphi})$ are collinear, we have
		\begin{subequations}
			\begin{align}
				\!\!\!\frac{ \cos^p(\vartheta - \varphi_1)}{ \cos^p(\theta_k - \varphi_1)} \!=\! \frac{ \cos^p(\vartheta - \varphi_2)}{ \cos^p(\theta_k - \varphi_2)}& \!=\! \dots \!=\!\frac{ \cos^p(\vartheta - \varphi_M)}{ \cos^p(\theta_k - \varphi_M)}\label{PD},\\
				\!\!\!\!\!\!|\theta_k -\varphi_m|<{\pi}/{2},&\enspace|\vartheta -\varphi_m|<{\pi}/{2}.
			\end{align}
		\end{subequations}
		Since the cosine function is strictly monotonically decreasing in the interval $(0, {\pi}/{2})$ and it is an even function, we can take the $p$-th root on both sides of the equalities in \eqref{PD} while maintaining their validity. Without loss of generality, we select the first and the $m$-th terms in \eqref{PD} to obtain
		\begin{align}
			\frac{\cos(\vartheta - \varphi_1)}{\cos(\vartheta - \varphi_m)} = \frac{\cos(\theta_k - \varphi_1)}{\cos(\theta_k - \varphi_m)},
		\end{align}
		where the equality applies the alternendo property of proportions. We then prove that the equality holds if and only if $\vartheta = \theta_k$, which is equivalent to proving that the ratio function $f(\vartheta) = \frac{\cos(\vartheta - \varphi_1)}{\cos(\vartheta - \varphi_m)}$ is injective. Taking the derivative of the ratio function yields
		\begin{align}
			f'(\vartheta) &= \frac{ -\sin(\vartheta - \varphi_1)\cos(\vartheta - \varphi_m) + \cos(\vartheta - \varphi_1)\sin(\vartheta - \varphi_m) }{ \cos^2(\vartheta - \varphi_m) }\nonumber\\
			& = \frac{ \sin\big[(\vartheta - \varphi_m) - (\vartheta - \varphi_1)\big] }{ \cos^2(\vartheta - \varphi_m) } = \frac{ \sin(\varphi_1 - \varphi_m) }{ \cos^2(\vartheta - \varphi_m) }.
		\end{align}
		Since $\varphi_m- \varphi_1 \le 2\vartheta_{\max}<\pi$, we have $f'(\vartheta)<0$, indicating that $f(\vartheta)$ is injective. Hence, $\mathcal{B}(\theta_k, \vartheta, \bm{\varphi})$ achieves the unique global maximum if and only if $\vartheta = \theta_k$. This completes the proof of Theorem \ref{Theorem1}.
	\end{proof}\vspace{-0.5em}
	Theorem \ref{Theorem1} demonstrates that when $M\ge2$, $\mathcal{B}(\theta_k, \vartheta, \bm{\varphi})$ achieves its unique maximum only at $\vartheta = \theta_k$. This implies that we can employ a correlation based algorithm to unambiguously estimate $\theta_k$ from the $k$-th column of the gain factor matrix $\bm{B}(\theta_k, \bm{\varphi})$. Based on this property, we can derive the following corollary.
	\begin{corollary}\label{Corollary1}
		$\mathcal{C}(\theta_k, \vartheta, L, \bm{\varphi}) \triangleq\mathcal{A}(\theta_k, \vartheta, L) \mathcal{B}(\theta_k, \vartheta, \bm{\varphi})$ has a unique maximum if and only if $\theta_k = \vartheta$.
	\end{corollary}
	\begin{proof}
		Specifically, $\mathcal{C}(\theta_k, \vartheta, L, \bm{\varphi})$ can be expressed as
		\begin{align}
			\!\!\!\!\!\!\!\!&\mathcal{C}(\theta_k, \! \vartheta, \! L,\!  \bm{\varphi}) \!=\!\frac{|\bm{a}_k^{\mathrm{H}}(\theta_k,\! L)\bm{a}(\vartheta, \!L)|^2}{\! ||\bm{a}_k(\theta_k,\! L)||_2^2  ||\bm{a}(\vartheta, \!L)||_2^2}\!\!  \frac{|\bm{b}_k^{\mathrm{T}}(\theta_k,\! \bm{\varphi})\bm{b}(\vartheta,\! \bm{\varphi})|^2}{||\bm{b}_k(\theta_k,\! \bm{\varphi})||_2^2 ||\bm{b}(\vartheta,\! \bm{\varphi})||_2^2}\nonumber\\
			&\!=\!\frac{|\big[\bm{a}_k(\theta_k, L) \otimes\bm{b}_k(\theta_k, \bm{\varphi})\big]^{\mathrm{H}} \big[\bm{a}(\vartheta, L) \otimes\bm{b}(\vartheta, \bm{\varphi})\big]|^2}{||\bm{a}_k(\theta_k, L) \otimes\bm{b}_k(\theta_k, \bm{\varphi})||_2^2 ||\bm{a}(\vartheta,L) \otimes\bm{b}(\vartheta, \bm{\varphi})||_2^2},
		\end{align}
		where the second equality applies the property of the Kronecker product $(\bm{a}\otimes\bm{b})^{\mathrm{H}}(\bm{c}\otimes\bm{d}) = (\bm{a}^{\mathrm{H}} \bm{c})(\bm{b}^{\mathrm{H}}\bm{d})$.
		According to the Cauchy-Schwarz inequality, the array SVC satisfies $\mathcal{A}(\theta_k, \vartheta, L) \in [0, 1]$ and the gain SVC satisfies $\mathcal{B}(\theta_k, \vartheta, \bm{\varphi}) \in [0, 1]$. Therefore, the joint SVC must satisfy $\mathcal{C}(\theta_k, \vartheta, L, \bm{\varphi}) \le 1$.
		To achieve the global maximum $\mathcal{C}(\theta_k, \vartheta, L, \bm{\varphi})= 1$ for the joint SVC, the following necessary conditions must be satisfied simultaneously, which requires $\mathcal{A}(\theta_k, \vartheta, L) = 1$ and $\mathcal{B}(\theta_k, \vartheta, \bm{\varphi}) = 1$.
		Due to the spatial undersampling nature of uniform sparse arrays, the solution set of $\mathcal{A}(\theta_k, \vartheta, L) = 1$ contains not only the true target angle $\theta_k$ but also multiple physically ambiguous grating lobe angles. However, according to Theorem \ref{Theorem1}, the necessary and sufficient condition for $\mathcal{B}(\theta_k, \vartheta, \bm{\varphi}) = 1$ to hold is exactly $\vartheta = \theta_k$.
		Thus, the solution to $\mathcal{C}(\theta_k, \vartheta, L, \bm{\varphi})=1$ is the intersection of the solutions for $\mathcal{A}(\theta_k, \vartheta, L) = 1$ and $\mathcal{B}(\theta_k, \vartheta, \bm{\varphi}) = 1$, yielding the unique solution $\vartheta = \theta_k$, which guarantees the uniqueness of multiple target estimation.
	\end{proof}\vspace{-0.5em}
	Corollary \ref{Corollary1} proves that although the array SVC inevitably suffers from grating lobe ambiguities caused by spatial undersampling, the joint SVC successfully suppresses all spurious peaks of non true targets because the gain SVC possesses a strict unimodal property that achieves its maximum if and only if $\vartheta = \theta_k$, thereby theoretically guaranteeing unambiguous DOA estimation. Furthermore, the joint SVC can fully exploit all information regarding the target DOA contained in the factor matrices $\bm{A}$ and $\bm{B}$, ensuring the algorithmic performance.
	
	We compute the Kronecker product of \eqref{aahat} and \eqref{bbhat} as
	\begin{align}
		\bm{\hat{c}}_k &= \bm{\hat{a}}_k \otimes \bm{\hat{b}}_k \in \mathbb{C}^{MN \times 1} \label{cchat}  \\
		&= \big[\delta_k^{(1)}\bm{a}_k(\theta_k, L)+\bm{e}_k^{(1)}\big] \otimes  \big[\delta_k^{(2)}\bm{b}_k(\theta_k, \bm{\varphi})+\bm{e}_k^{(2)}\big]\nonumber\\
		& = \delta_k^{(1,2)}\bm{a}_k(\theta_k, L) \otimes \bm{b}_k(\theta_k, \bm{\varphi}) +\bm{e}_k^{(1,2)}\nonumber,
	\end{align}
	where $\delta_k^{(1,2)} = \delta_k^{(1)}\delta_k^{(2)}$ and $\bm{c}_k (\theta_k,L, \bm{\varphi})= \bm{a}_k(\theta_k, L) \otimes \bm{b}_k(\theta_k, \bm{\varphi})$, and the specific expression of $\bm{e}_{k}^{(1,2)}\in \mathbb{C}^{MN \times 1}$ is given by
	\begin{align}
		\!\!\!\!\bm{e}_{k}^{(1,\!2)}\! \!=\! \delta_k^{(1)}\! \bm{a}_k(\theta_k,\! L) \!\otimes\! \bm{e}_k^{(2)}\! \!+\! \delta_k^{(2)}\! \bm{e}_k^{(1)}\!\otimes\! \bm{b}_k(\theta_k,\! \bm{\varphi}) \!+\! \bm{e}_k^{(1)} \!\otimes\!\bm{e}_k^{(2)}\!\!.
	\end{align}
	According to Corollary \ref{Corollary1}, we apply a correlation-based search algorithm to \eqref{cchat} to unambiguously estimate the DOA of the $k$-th target $\hat{\theta}_k$ as
	\begin{align}\label{ccSVCg}
		\hat{\theta}_k &= \arg\max_{\vartheta} \frac{|\bm{\hat{c}}_k^{\mathrm{H}} \bm{c}(\vartheta, L, \bm{\varphi})|^2}{||\bm{\hat{c}}_k||_2^2 ||\bm{c}( \vartheta, L, \bm{\varphi})||_2^2 }\\
		& = \arg\max_{\vartheta} \frac{|\big(\bm{\hat{a}}_k \otimes\bm{\hat{b}}_k\big)^{\mathrm{H}}
			\big(\bm{a}(\vartheta, L) \otimes\bm{b}(\vartheta, \bm{\varphi})\big)|^2}{||\bm{\hat{a}}_k \otimes  \bm{\hat{b}}_k||_2^2|| \bm{a}(\vartheta, L) \otimes \bm{b}(\vartheta, \bm{\varphi}) ||_2^2 }.\nonumber
	\end{align}
	To this end, we have completed the DOA estimation via the proposed RA-enhanced tensor decomposition-based algorithm.
	
	%	\begin{algorithm}[t]
		%		\renewcommand{\algorithmicrequire}{\textbf{Input }}  %Use Input in the format of Algorithm
		%		\renewcommand{\algorithmicensure}{\textbf{Output }}  %UseOutput in the format of Algorithm
		%		\caption{Proposed AO-Based Self-Calibration Algorithm}
		%		\label{Algorithm 1}
		%		\begin{algorithmic}[1]
			%			\Require 
			%			The received signal $\bm{Z}$, the number of snapshots $T$, the number of signal sources $K$, the movable region $H$, the number of MAs $M$, and the number of calibrated MAs $M_c$.
			%			\State Calculate the noise subspcace $\bm{E}_N$ via \eqref{RY} and \eqref{EVD};
			%			\State Construct the constraint condition via \eqref{condition};
			%			\While {$l\leq L$ or $\|\bm{\hat{\theta}}^{(l)}-\bm{\hat{\theta}}^{(l-1)}\|_{\mathrm{F}}^2\geq \delta$}
			%			
			%			\Statex {\textbf{Stage 1: DOA Estimation}}
			%			\State Calculate the estimated DOA $\bm{\theta}^{(l)}$ via \eqref{Initial}.
			%			\Statex {\textbf{Stage 2: APE Estimation}}
			%			\For {each $k\in [1,\ldots,K]$ }
			%			
			%			\State Estimate $\bm{\hat{{a}}}(\theta_k,{\Delta\bm{P}})$ via \eqref{ahat};
			%			\State Compute the function via  \eqref{only};
			%			\EndFor
			%			\State Estimate the APE of MA $[{\Delta \bm{\hat{P}}}^{(l)}]_{m,:}$ via \eqref{all};
			%			\State Update the estimated actual positions of MA $\bm{\hat{{P}}}^{(l-1)}$ via $\bm{\hat{ P}}^{(l)}={\Delta \bm{\hat{P}}}^{(l)}+\bm{\tilde{P}}$;
			%			\State $l\leftarrow l+1$;
			%			\EndWhile
			%			\Ensure Return estimated parameters $\bm{\hat{\theta}}$ and $\bm{\hat{ P}}$.
			%		\end{algorithmic}
		%		%\vspace{-1em}
		%	\end{algorithm}

	%\vspace{-0.5em}
	\section{Simulation Results}
	In this section, we evaluate the performance of the proposed algorithm through numerical simulations. The signal-to-noise ratio (SNR) is given by $\mathrm{SNR} = 10\lg \frac{\sigma_k^2}{\sigma_n^2}$. %and root mean square error (RMSE) is calculated by $\mathrm{RMSE} = \sqrt{\frac{1}{QK}\sum_{q=1}^{Q}{\sum_{k=1}^{K}{\big(\theta_k - \hat{\theta}_{k,q}\big)^2}}}$ with $Q$ and $\hat{\theta}_{k,q}$ denoting total number of the Monte Carlo simulation and the estimated DOA of the $k$-th target during the $q$-th simulation. 
	Without loss of generality, we normalize the scattering coefficients as $\alpha_k=1$ and set the signal power as $\sigma_k^2 =1$ for $k =1,\ldots,K$. Unless otherwise specified, the default parameters are set as follows. The sensing range is evaluated within $[-\pi/3, \pi/3]$ with three target DOAs located at $[-20^\circ, 15^\circ, 45^\circ]$. The receiver is equipped with $N = 8$ RAs operating under a sparse factor of $L = 2$ and a directivity factor of $p = 3$. During the active probing phase, the RAs perform $M = 7$ synchronous rotations. Furthermore, we collect $T = 20$ snapshots for each observation under an SNR of 10 $\mathrm{dB}$.  For performance comparison, we consider the following baseline schemes:
	\begin{itemize}
		\item \textbf{Proposed US-RA} which denotes the proposed algorithm equipped with RAs under the USA configuration.
		\item \textbf{Proposed UD-RA}, which represents the proposed algorithm equipped with RAs under the UDA configuration.
		\item \textbf{Conventional US-OA}, which refers to the conventional MUSIC algorithm equipped with OAs under the USA configuration\footnote{It is worth noting that the received signals in conventional schemes cannot be formulated as a tensor due to the absence of antenna rotation. To guarantee a fair comparison with an identical volume of received data, we aggregate the snapshots corresponding to the multiple rotation periods of RAs for the conventional schemes.}.
		\item \textbf{Conventional UD-OA}, which refers to the conventional MUSIC algorithm equipped with OAs under the UDA configuration.
	\end{itemize}
	
	\begin{figure}[t]
		\centering
		\begin{minipage}{0.43\linewidth}
			\centering
			\includegraphics[width=\linewidth]{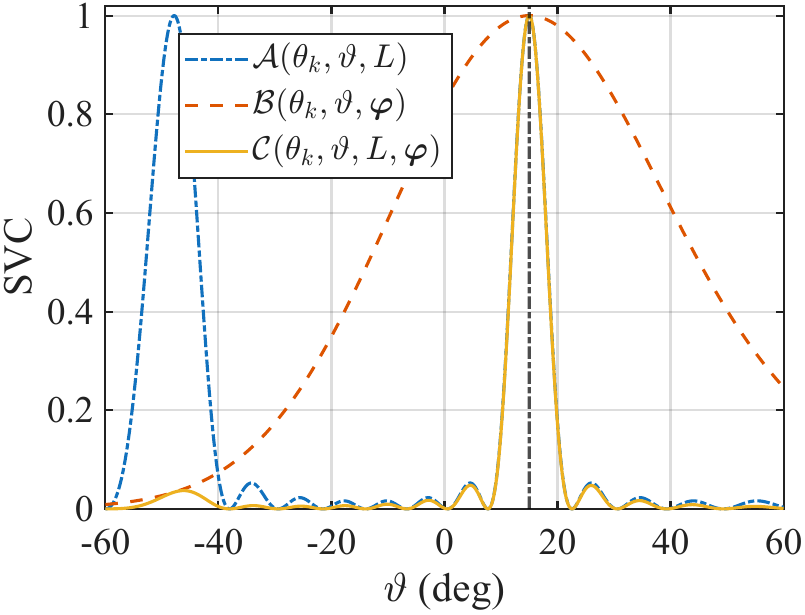}
			\caption{SVC versus $\vartheta$.}
			\label{SVC_ABC_Proposed_RA}
		\end{minipage}
		\hfill
		\begin{minipage}{0.43\linewidth}
			\centering
			\includegraphics[width=\linewidth]{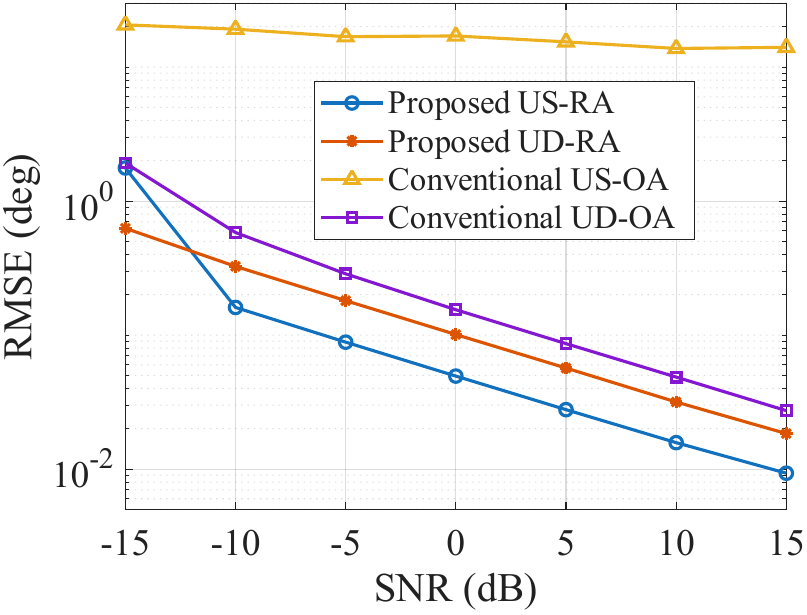}
			\caption{RMSE versus SNR.}
			\label{RA_RMSE_SNR}
			
		\end{minipage}
		\vspace{-0.75em}
	\end{figure}
	\begin{figure}[t]
		\centering
		\begin{minipage}{0.43\linewidth}
			\centering
			\includegraphics[width=\linewidth]{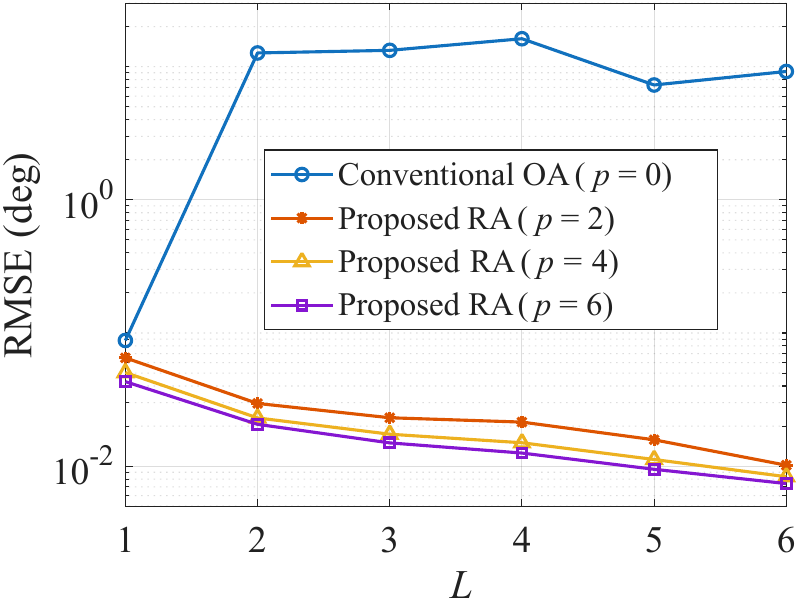}
			\caption{RMSE versus the sparse factor $L$.}
			\label{RA_RMSE_d_lambda}
		\end{minipage}
		\hfill
		\begin{minipage}{0.43\linewidth}
			\centering
			\includegraphics[width=\linewidth]{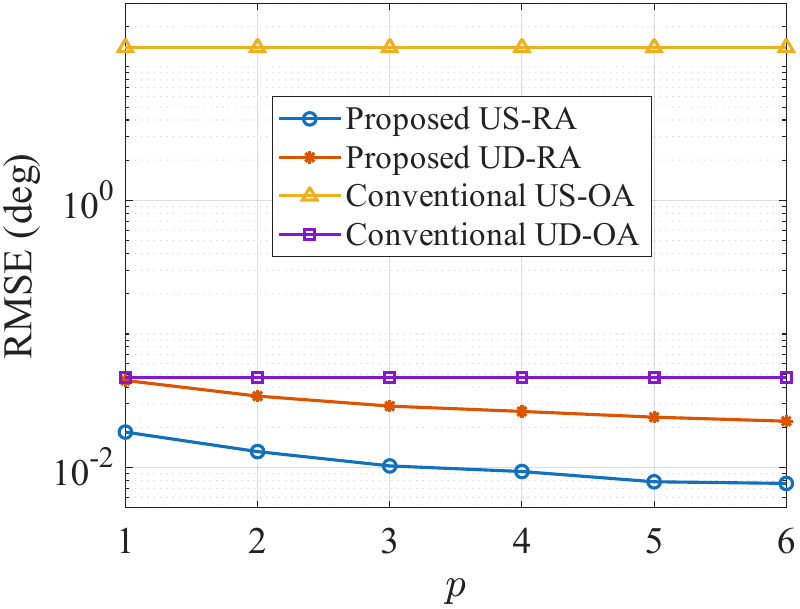}
			\caption{RMSE versus the directivity factor $p$.}
			\label{RA_RMSE_P}
		\end{minipage}
		\vspace{-0.75em}
	\end{figure}
	
	\Cref{SVC_ABC_Proposed_RA} depicts the variations of the array SVC $\mathcal{A}(\theta_k, \vartheta, L)$, the gain SVC $\mathcal{B}(\theta_k, \vartheta, \bm{\varphi})$, and the joint SVC $\mathcal{C}(\theta_k, \vartheta, \bm{\varphi}, L)$ with respect to the sensing angle. In this simulation, the true DOA of the $k$-th target $\theta_k = 15^\circ$ and the directivity factor is set to $p = 5$. From the simulation results, we can observe that $\mathcal{A}(\theta_k, \vartheta, L)$ exhibits severe grating lobes due to the USA configuration of RAs, which makes unambiguous DOA estimation impossible. However, $\mathcal{B}(\theta_k, \vartheta, \bm{\varphi})$ possesses a unique main lobe within the entire sensing range and exhibits extremely small values at the spatial directions corresponding to the grating lobes of $\mathcal{A}(\theta_k, \vartheta, L)$. Consequently, the grating lobes in $\mathcal{C}(\theta_k, \vartheta, \bm{\varphi}, L)$ are significantly suppressed into negligible sidelobes to enable unambiguous DOA estimation. This observation is very consistent with Theorem \ref{Theorem1} and Corollary \ref{Corollary1}.
	
	\Cref{RA_RMSE_SNR} compares the RMSE performance of different baseline schemes against the SNR. We can observe that the conventional US-OA scheme fails to estimate the DOAs due to the severe grating lobe issue. Furthermore, both the proposed UD-RA and US-RA schemes significantly outperform the conventional UD-OA scheme. On one hand, the employed RAs can concentrate the signal energy within a specific angular range to provide a larger received power gain compared to the OAs. On the other hand, the USA configuration effectively expands the spatial array aperture to further enhance the overall sensing performance.
	
	\Cref{RA_RMSE_d_lambda} presents the evaluation of the RMSE with respect to the sparse factor $L$. The SNR is set to 5 $\mathrm{dB}$ in this simulation. We can observe that the conventional OA scheme fails to estimate the DOAs when $L > 1$ due to the severe grating lobe issues, despite the enlarged array aperture. Conversely, the performance of the proposed algorithm continuously improves with the increase of $L$ under different directivity factors of $p = 2$, $p = 4$, and $p = 6$. Furthermore, a larger directivity factor $p$ yields better DOA estimation performance.
	
	\Cref{RA_RMSE_P} investigates the impact of the RA directivity factor $p$ on the RMSE performance. We can observe that for a fixed sparse factor, the sensing performance of both the UDA and USA  schemes gradually improves as $p$ increases. Furthermore, the USA scheme consistently outperforms the UDA scheme across the entire range of $p$.

	\vspace{-0.5em}
	
	\section{Conclusion}
	In this letter, we proposed a novel tensor decomposition-based DOA estimation algorithm for RA-enhanced USA sensing systems to tackle spatial undersampling issue. By formulating the received signals via successive rotations as a third-order tensor, we analytically demonstrate that the strict unimodality of the gain SVC effectively suppresses the inherent grating lobes of the array SVC. By leveraging this unique property, we combine the array and gain factor matrices via the Kronecker product to ensure an unambiguous DOA estimation, which achieves robust and high precision sensing performance that significantly surpasses that of conventional UDA and OA systems.
	
	\vspace{-0.5em}

	
	
	%\bibliographystyle{IEEEtran}
	%\FloatBarrier
	%\vspace{-0.5em}
	\raggedbottom 
	\bibliographystyle{IEEEtran}
	\bibliography{mybib_Abbreviation_1}

\end{document}